\documentclass[12pt,letterpaper]{article}
\usepackage{amsfonts}
\usepackage{amsmath}
\usepackage{geometry}
\usepackage{amsthm}
\usepackage{bbm}
\usepackage[authoryear,round,longnamesfirst]{natbib}
\usepackage{xcolor}
\usepackage{grffile}
\usepackage{graphicx}

\theoremstyle{definition}
\newtheorem{example}{Example}
\theoremstyle{plain}
\newtheorem{assumption}{Assumption}
\newtheorem{theorem}{Theorem}

\newtheorem{proposition}{Proposition}
\newtheorem{corollary}{Corollary}

\allowdisplaybreaks

\begin{document}

\title{{\Large Extreme Quantile Treatment Effects under Endogeneity: Evaluating Policy Effects for the Most Vulnerable Individuals}}
\author{ Yuya Sasaki\thanks{\setlength{\baselineskip}{4.4mm}
Brian and Charlotte Grove Chair and Professor of Economics, Vanderbilt University. Email:
yuya.sasaki@vanderbilt.edu.} \ and \ Yulong Wang\thanks{%
Associate Professor of Economics, Syracuse University. Email:
ywang402@syr.edu.}}
\date{}
\maketitle

\begin{abstract}
\setlength{\baselineskip}{8.25mm}
We introduce a novel method for estimating and conducting inference about \textit{extreme} quantile treatment effects (QTEs) in the presence of endogeneity. 
Our approach is applicable to a broad range of empirical research designs, including instrumental variables design and regression discontinuity design, among others. 
By leveraging regular variation and subsampling, the method ensures robust performance even in extreme tails, where data may be sparse or entirely absent.
Simulation studies confirm the theoretical robustness of our approach.
Applying our method to assess the impact of job training provided by the Job Training Partnership Act (JTPA), we find significantly negative QTEs for the lowest quantiles (i.e., the most disadvantaged individuals), contrasting with previous literature that emphasizes positive QTEs for intermediate quantiles.

{\small {\ \ \ \newline
} \newline
}
\end{abstract}

\newpage


\section{Introduction}

\label{sec:intro} 

The quantile treatment effect (QTE) is a widely adopted concept in empirical research for quantifying heterogeneous treatment effects. Various methods have been proposed to identify the QTE in the presence of endogeneity. Popular approaches include instrumental variables designs \citep{abadie2002instrumental,chernozhukov2005iv}, changes-in-changes designs \citep{athey2006identification}, and regression discontinuity designs \citep{frandsen2012quantile}, among others. Their identification results are often accompanied by corresponding methods for estimation and inference.

However, existing methods of estimation and inference are predominantly developed for intermediate quantiles, leaving a gap in the literature regarding estimation and inference for extreme quantiles, such as those at the very low or very high ends of the distribution. This is a notable limitation from practical perspectives, as policymakers are often particularly interested in the effectiveness of policies for extreme subpopulations, such as individuals living in extreme poverty or those with severe health conditions.

Two studies address this issue in specific contexts: \citet{zhang2018extremal} focuses on the observed unconfoundedness design; and \citet{sasaki2024extreme} focus on changes-in-changes design. However, extreme QTEs under other endogeneity designs remain unexamined. This paper aims to fill this gap by introducing a novel method for the estimation and inference of extreme QTEs across a broad class of endogeneity scenarios and research designs.

Consider situations where the distributions, $F_1$ and $F_0$, of the potential outcomes under treatment and no treatment, respectively, are identifiable. In these cases, the QTE is quantified by $F_1^{-1}(q) - F_0^{-1}(q) $ for $q \in (0,1)$, where $F_j^{-1}$ denotes the left inverse of $F_j$ for each $j \in \{0,1\}$.
Our proposed method is applicable in situations where $F_1$ and $F_0$ can be identified and estimated by $\hat{F}_1$ and $\hat{F}_0$, respectively, with $r_n(\hat{F}_j(\cdot) - F_j(\cdot)), \, j \in \{0,1\} $, weakly converging to a Gaussian process at some rate $r_n$. All the aforementioned examples, along with many others, satisfy this requirement, making our proposed method applicable to a broad class of designs.

Instead of directly using $\hat{F}_1^{-1}(q) - \hat{F}_0^{-1}(q)$ as the estimator, we first estimate the tail indices of $F_1$ and $F_0$, and then obtain estimates of the QTE through Pareto approximation, which is valid under the regular variation condition. 
This approach enables us to robustly estimate the QTE even in regions where there are few or even no observations in the tails. 
Subsequently, we propose a method of subsampling inference for the extreme QTE based on this estimation strategy.

We demonstrate the applicability of our method through two leading examples under endogeneity: the instrumental variables design \citep{abadie2003semiparametric} and the regression discontinuity design \citep{frandsen2012quantile}. Our simulation studies reveal that the proposed method performs robustly in both scenarios.

We present an empirical application of the proposed method to evaluate the effects of job training provided by the Job Training Partnership Act (JTPA). Previous literature has highlighted the positive QTE of this program in the intermediate quantiles. However, our study of the extremely low quantiles reveals rather negative treatment effects, which are statistically significant. This finding suggests that the program may have adverse outcomes for the most disadvantaged individuals.

\bigskip\noindent 
\textbf{Literature:} 
There is a substantial body of literature in statistics and econometrics addressing extreme quantiles -- see the book by \citet{de2006extreme}.
\citet{chernozhukov2005extremal} proposes methods for estimation and inference in quantile regressions, focusing on the tails of the outcome distribution. 
\citet{chernozhukov2011inference} introduce a subsampling approach for inference in these settings.
For quantile treatment effects (QTEs), \citet{d2018extremal} consider a sample selection model, exploiting tail independence in selection. 
\citet{zhang2018extremal} analyzes estimation and inference for extreme QTEs within a broad class of observational data characterized by observational unconfoundedness.
To our knowledge, only one paper examines extreme QTEs under endogeneity: \citet{sasaki2024extreme} investigate this within the changes-in-changes (CIC) framework. Their method capitalizes on the unique features of the identifying formula in the CIC, but it does not generalize to other models with endogeneity, such as those using instrumental variables or regression discontinuity designs, which are the focus of this paper.
See \citet{chernozhukov2017extremal} for a survey.



\section{The Main Theoretical Result}

\label{sec:main} 

Let $Y_0$ and $Y_1$ denote the potential outcomes without and with treatment, respectively. 
Let $F_{j}\left( \cdot \right)$ denote a suitable (conditional) CDF of $Y_{j}$ of interest for each $j \in \{0,1\}$. 
Suppose that $\beta_j\left( \cdot \right)$ identifies $F_j\left( \cdot \right)$, i.e., 
\begin{equation}  \label{eq:beta}
F_{j}\left( \cdot \right) = \beta_{j}\left( \cdot \right).
\end{equation}
In Examples \ref{ex:abadie} and \ref{ex:frandsen} below, we provide a couple of examples of $\beta_j$ drawn from the literature on identification. 
For each of these examples, there is a plug-in analog estimator $\hat\beta_j$ of $\beta_j$.

For intermediate quantiles $q$, the $q$-th quantile treatment effect, 
\begin{equation*}
F_{1}^{-1}\left( q \right) -F_{0}^{-1}\left( q \right), 
\end{equation*}
can be estimated by the sample counterpart $\hat\beta_{1}^{-1}\left( q \right) - \hat\beta_{0}^{-1}\left( q \right), $ where $\hat\beta_{j}^{-1}(\cdot)$ denotes the left inverse of $\hat\beta_{j}(\cdot)$. 
On the other hand, for extreme quantiles where $ q $ is close to zero or one, this na\"ive estimator behaves poorly.

Our goal is estimate the extreme quantile $F_{j}^{-1}\left( q \right)$ for $ q $ close to zero or one. 
Without loss of generality, we consider the right tail, that is, $q \rightarrow 1$. 
The extreme quantiles at the left tail can be analyzed by symmetric arguments.

By imposing regularly varying tail on $F_{j}$, we show that we can robustly estimate the extreme quantiles $F_{j}^{-1}\left(q \right)$. 
Consequently, the extreme QTE is estimated precisely. 
To this goal, we make the following assumption on the estimator $\hat\beta_{j}(\cdot)$ of $\beta_{j}\left( \cdot \right) = F_{j}\left( \cdot \right).$

\begin{assumption}
\label{a:beta_hat} For each $j \in \{0,1\}$, there exist sequences of constants $r_{n} \rightarrow \infty $ as $n \rightarrow \infty$ such that 
\begin{align*}
&r_{n} \left(\begin{array}{c}  \hat{\beta}_{1}\left( \cdot \right) -\beta_{1}\left( \cdot \right)   \\ \hat{\beta}_{0}\left( \cdot \right) -\beta_{0}\left( \cdot \right)  \end{array}\right) \Rightarrow \left(\begin{array}{c} Z_{1}\left( \cdot \right)  \\ Z_{0}\left( \cdot \right) \end{array}\right) ,
\end{align*}
for some Gaussian processes $(Z_{1}\left( \cdot \right), Z_{0}\left( \cdot \right))$ with covariance function 
\begin{equation*}
\Xi \left( \cdot, \cdot \right) =  \left(\begin{array}{cc} \Xi_{1}\left( \cdot, \cdot \right) \text{ , }  \Xi_{10} \left( \cdot, \cdot \right)     \\  \Xi_{10}\left( \cdot, \cdot \right)  \text{ , }   \Xi_{0} \left( \cdot, \cdot \right)       \end{array}\right) .
\end{equation*}
\end{assumption}

Assumption \ref{a:beta_hat} can be and has been shown in the literature to be satisfied under a variety of settings, as illustrated in the two examples below.

\begin{example}[Binary Instrument and Binary Treatment; Abadie, 2003] \label{ex:abadie} 
Suppose that an econometrician observes $n$ copies of $(Y,D,Z,X^{\prime })$, where $Y$ is an outcome, $D$ is a binary treatment indicator, $Z$ is a binary instrument, and $X$ is a vector of covariates.
Let $p(X) = P(Z=1|X)$. The conditional distribution $F_0(y) = E[1\{Y_0 \le y\} | D_1 > D_0]$ of the potential outcome $Y_0$ given compliers $D_1 > D_0$ is identified by 
\begin{equation*}
\beta_0(y) = \frac{ E\left[ 1\{ Y<y \} \cdot \frac{(1-D)(p(X)-Z)}{(1-p(X))p(X)}\right] }{ E\left[1 - \frac{D(1-Z)}{1-p(X)} - \frac{(1-D)Z}{p(X)}\right] }. 
\end{equation*}
Similarly, the conditional distribution $F_1(y) = E[1\{Y_1 \le y\} | D_1 > D_0]$ of the potential outcome $Y_1$ given compliers $D_1 > D_0$ is identified by 
\begin{equation*}
\beta_1(y) = \frac{ E\left[ 1\{ Y<y \} \cdot \frac{D(Z-p(X))}{(1-p(X))p(X)} \right] }{ E\left[1 - \frac{D(1-Z)}{1-p(X)} - \frac{(1-D)Z}{p(X)}\right] } 
\end{equation*}
If we consider a parametric estimator $\widehat p(X)$ for $p(X)$, then sample-counterpart estimator $\widehat\beta_j(\cdot)$ with the plug-in of $ \widehat p(X)$ satisfies 
\begin{equation}\label{eq:weak_abadie}
\sqrt{n}\left(\widehat\beta_j( \cdot ) - \beta_j( \cdot )\right) \Rightarrow Z_j(\cdot),
\end{equation}
$j \in \{0,1\}$. Appendix \ref{sec:abadie_weak_convergence} formally establishes conditions under which the weak convergence \eqref{eq:weak_abadie} holds.
\qed
\end{example}

\begin{example}[Regression Discontinuity Design; Frandsen, Frolich, and Melly, 2012] \label{ex:frandsen} 
Suppose that an econometrician observes $n$ copies of $(Y,D,R)$, where $Y$ is an outcome, $D$ is a binary treatment indicator, and $R$ is a running variable. 
The conditional distribution $F_0(y) = E[1\{Y_0 \le y\} | D_1 > D_0, R=0]$ of the potential outcome $Y_0$ given local compliers $D_1 > D_0$ at $R=0$ is identified by 
\begin{equation*}
\beta_0(y) = \frac{\lim_{r \downarrow 0}E[1\{Y \le y\}(1-D)|R=r] - \lim_{r \uparrow 0}E[1\{Y \le y\}(1-D)|R=r]}{\lim_{r \downarrow 0}E[(1-D)|R=r] - \lim_{r \uparrow 0}E[(1-D)|R=r]}. 
\end{equation*}
Similarly, the conditional distribution $F_0(y) = E[1\{Y_0 \le y\} | D_1 > D_0, R=0]$ of the potential outcome $Y_1$ given local compliers $D_1 > D_0$ at $R=0$ is identified by 
\begin{equation*}
\beta_1(y) = \frac{\lim_{r \downarrow 0}E[1\{Y \le y\}D|R=r] - \lim_{r \uparrow 0}E[1\{Y \le y\}D|R=r]}{\lim_{r \downarrow 0}E[D|R=r] - \lim_{r \uparrow 0}E[D|R=r]}. 
\end{equation*}
\citet[][Theorem 2]{frandsen2012quantile} also propose sample-counterpart estimator $\hat\beta_j$ and show 
\begin{equation}\label{eq:weak_frandsen}
\sqrt{nh}\left(\widehat\beta_j( \cdot ) - \beta_j( \cdot )\right) \Rightarrow Z_j(\cdot),
\end{equation}
$j \in \{0,1\}$. \citet[][Theorem 2]{frandsen2012quantile} formally establish conditions under which the weak convergence \eqref{eq:weak_frandsen} holds.
\qed
\end{example}

Given our focus on the tail, we also need the following additional condition to regularize the tail of the CDF estimator $\hat{\beta}_j(\cdot)$. 

\begin{assumption}
\label{a:xi} 
For each $j \in \{0,1\}$, there exist sequences of constants $ \tilde{r}_n \rightarrow \infty $ as $n \rightarrow \infty$ such that 
\begin{equation*}
\tilde{r}_{n}  \sup_{y \in \mathbb{R}} \left\vert \hat{\beta}_{j}\left( y \right) -\beta _{j}\left( y \right) \right\vert  \Xi_{j}(y,y)^{-1/2} = O_p(1),
\end{equation*}
where the covariance function $\Xi _{j}\left( \cdot, \cdot \right)$ is defined as in Assumption \ref{a:beta_hat}.
\end{assumption}

Assumption \ref{a:xi} is arguably mild. 
Compared with Assumption \ref{a:beta_hat}, it essentially assumes that the normalized t-statistic is uniformly bounded in probability. 
We present this assumption in its current form for generality, and primitive conditions can be derived case-by-case. 

We now move on to learn about the tail features of $F_{j}\left( \cdot \right) $. 
To this end, we assume that $F_{j}\left( \cdot \right) $ is regularly varying (RV) at infinity.
We first provide an intuitive illustration to be followed by a formal theory.
Specifically, suppose that $F_j(\cdot)$ satisfies 
\begin{equation}
\lim_{t\rightarrow \infty }\frac{1-F_{j}\left( yt\right) }{1-F_{j}\left( t\right) }=y^{-\alpha _{j}}  \label{eq:RV}
\end{equation}%
for any $y>0$, where the constant $\alpha _{j}$ is called the Pareto exponent, which characterizes the tail heaviness of $F_{j}\left( \cdot \right) $. 
The RV condition allows for the approximation 
\begin{equation}
\frac{1-F_{j}\left( y\right) }{1-F_{j}\left( y_{\min,j}\right) }\approx \left( \frac{y}{y_{\min,j}}\right) ^{-\alpha _{j}}  \label{eq:RV approx}
\end{equation}%
for all values $y>y_{\min,j}$ for sufficiently large $y_{\min,j}$. In practice, $y_{\min,j}$ plays the role of a tuning parameter and we will discuss the choice of $y_{\min,j} $ later.

Since $\hat{\beta}_{j}\left( \cdot \right) $ consistently estimates $F_{j}\left( \cdot \right) $, we can estimate $\alpha _{j}$ by fitting $\hat{\beta}_{j}\left( \cdot \right) $ with the Pareto distribution. 
By \eqref{eq:beta} and \eqref{eq:RV approx}, we have
\begin{equation*}
\log \left( \frac{1-\beta _{j}\left( y\right) }{1-\beta _{j}\left(y_{\min,j}\right) }\right) \approx -\alpha _{j}\log \left( \frac{y}{y_{\min,j}}\right)
\end{equation*}
for $y>y_{\min,j}$. 
Integrating both sides with respect to $y$ motivates the estimator 
\begin{equation}
\hat{\alpha}_{j}=-\frac{\int_{y_{\min,j}}^{\infty}\log \left( \frac{1-\hat{\beta}_{j}\left( y\right) }{1-\hat{\beta}_{j}\left( y_{\min,j}\right) }\right) w_j\left( y\right) dy}{\int_{y_{\min,j}}^{\infty}\log \left( \frac{y}{y_{\min,j}}\right) w_j\left( y\right) dy}  \label{eq:alpha hat}
\end{equation}
of the Pareto exponent $\alpha _{j}$, where $w\left( \cdot \right) $ is some weighting function. This weighting function serves as regularization so that integrands in the numerator and the denominator are integrable. 

It is natural to choose $w_j\left( \cdot \right) $ as the true density of the
Pareto distribution with exponent $\alpha _{j}$, that is, $w_j\left( y\right) ={y^{-\alpha _{j}-1}}/{y_{\min,j}^{-\alpha _{j}}}.$ 
However, this is infeasible as $\alpha _{j}$ is unknown. 
As an alternative, we can choose any $\omega >0$ and set 
\begin{equation}
w_j\left( y\right) ={y^{-\omega -1}}/{y_{\min,j}^{-\omega }}.  \label{eq:w}
\end{equation}
In this case, the denominator of \eqref{eq:alpha hat} reduces to 
\begin{equation}
\int_{y_{\min,j}}^{\infty}\log \left( \frac{y}{y_{\min,j}}\right) \frac{y^{-\omega -1}}{y_{\min,j}^{-\omega }}dy=\int_{1}^{\infty}\log \left( s\right) s^{-\omega -1}dy\rightarrow \frac{1}{\omega ^{2}}. \label{eq:w2}
\end{equation}

To derive the asymptotic distribution of $\hat{\alpha}_{j}$, a second-order approximation is inevitable besides \eqref{eq:RV}. 
Specifically, we impose the following condition to formally characterize \eqref{eq:RV approx}.

\begin{assumption}
\label{a:pareto} For each $j\in \{0,1\}$, $F_j(\cdot)$ satisfies 
\begin{equation*}
1-F_{j}\left( y\right) =c_{j}y^{-\alpha_{j}}(1+d_{j}y^{-\rho_{j}}+o(y^{-\rho _{j}})) \text{ as } y\rightarrow \infty,
\end{equation*}
for some constants $\alpha _{j}>0$, $\rho _{j}>0$, $c_{j}>0$, and $d_{j}\in\mathbb{R}$.
\end{assumption}

First, we remark that this assumption is mild and satisfied by all the commonly used families of heavy-tailed distributions such as Student-t, F, Cauchy, etc. 
In particular, for the standard Student-t distribution with $v$ degrees of freedom, it holds with $\alpha _{j}=v$ and $\rho _{j}=2$. 
See, for example, \citet[][p.62]{de2006extreme}.

Second, Assumption \ref{a:pareto} is used to control the leading bias of the estimator \eqref{eq:alpha hat}. 
It is, therefore, analogous to the condition of bounded second (or higher) order derivatives in the context of kernel estimation.
The parameters, $d_j$ and $\rho_j$, measure deviation of $F_j(\cdot)$ away from the reference Pareto distribution with exponent $\alpha_j$. 

To obtain limiting distribution for the estimator $\hat{\alpha}_j$ around a pseudo true value, it suffices to choose $y_{\min,j}$ according to the following assumption.

\begin{assumption}\label{a:y_bar} 
As $n\rightarrow \infty $, $y_{\min,j}\rightarrow \infty $, and $y_{\min,j}^{\alpha _{j}}\Xi(y_{\min,j}, y_{\min,j})^{1/2}/\tilde{r}_{n}\rightarrow 0$ for each $j\in \{0,1\}$, where $\tilde{r}_n$ is given in Assumption \ref{a:beta_hat}. 
\end{assumption}

Assumption \ref{a:y_bar} concerns limiting behaviors of the regularization parameter $y_{\min,j}$ as $n \rightarrow \infty$.
It controls the bias due to the deviation of $F_j$ from the benchmark Pareto distribution, which converges at the rate of $y^{-\rho_j}$ as $y \rightarrow \infty$.

The following theorem establishes the asymptotic distribution of $\hat{\alpha}_j$.

\begin{theorem}
\label{thm:index} Suppose Assumptions \ref{a:beta_hat} to \ref{a:y_bar} hold. With the weighting function \eqref{eq:w}, it holds that for each $j\in \{0,1\}$
\begin{equation*}
A_{jn}^{-1/2}\left( \hat{\alpha}_{j}-\alpha _{j}+B_{jn}\right) \overset{d}{\rightarrow }\mathcal{N}\left( 0,1\right), 
\end{equation*}
where 
\begin{align*}
A_{jn}=&\; \omega ^{4}Var\left[ r_{n}^{-1}\int_{y_{\min,j}}^{\infty}\left( \frac{Z\left( y\right) }{c_jy^{-\alpha _{j}}}-\frac{Z\left( y_{\min,j}\right) }{c_jy_{\min,j}^{-\alpha _{j}}}\right) w_j\left( y\right) dy\right] \qquad \text{and} \\
B_{jn}=&\; d_{j}\omega ^{2}y_{\min,j}^{-\rho _{j}}\left( \frac{1}{\rho_{j}+\omega }-\frac{1}{\omega }\right) .
\end{align*}
\end{theorem}

\noindent A proof is found in Appendix \ref{sec:proof:index}.

Theorem \ref{thm:index} derives the asymptotic normality of $\hat{\alpha}_{j}$. 
The $A_{jn}$ and $B_{jn}$ terms respectively characterize the variance and the bias. 
One can control the balance between the bias and variance via the tuning parameter $y_{\min,j}$. 
A larger $y_{\min,j}$ leads to smaller bias and larger variance. 
The theoretically optimal choice $y_{\min,j}^{\ast}$ requires $A_{jn}^{-1/2}\times B_{jn} \sim 1$. 
If we further assume that the variance $\Xi_j(y, y)$ of $Z_j(y)$ is approximately proportional to $y^{-\kappa_j}$ for some $\kappa_j>0$ in the limit as $y \rightarrow \infty$, then we can derive the optimal choice as $y_{\min,j}^{\ast} \sim r_{n}^{1/\left( \alpha _{j}-\kappa_{j}/2-1/2+\rho_{j}\right) }$. 
Since the second-order parameter $\rho _{j}$ is challenging to estimate in practice and the bias should vanish faster for the purpose of inference, we recommend setting $y_{\min,j}^{\ast} /r_{n}^{1/\left(\alpha_{j}-\kappa _{j}/2-1/2+\rho _{j}\right) }\rightarrow \infty $ so that $A_{jn}^{-1/2}\times B_{jn}=o(1)$. 
In other words, $y_{\min,j}^{\ast}$ is an undersmoothing choice.
With this said, such a rule must be subject to a feasible choice of $y_{\min,j}$ satisfying Assumption \ref{a:y_bar}, i.e., $r_{n}^{1/\left(\alpha_{j}-\kappa _{j}/2-1/2+\rho _{j}\right) } \prec y_{\min,j} \prec \tilde{r}_n^{1/(\alpha_j - \kappa_{j}/2)}$, which further requires that $\rho_j>1/2$. 
Note that a smaller $\rho_j$ (closer to zero) implies a larger deviation from the benchmark Pareto distribution.
The implicit requirement, $\rho_j>1/2$, for a feasible choice of $y_{\min,j}$, therefore, rules out huge deviations away (i.e., $\rho_j$ close to zero) from the benchmark Pareto distribution.\footnote{Student-t distribution entails $\rho_j = 2$ and hence satisfies this restriction.}

To conduct valid inference, we still need to know the constant $c_{j}$ in $A_{jn}$ as specified in Assumption \ref{a:pareto}, which can be consistently estimated by 
\begin{equation*}
\hat{c}_{j}=\left( 1-\hat{\beta}_{j}\left( y_{\min,j}\right) \right) y_{\min,j}^{\hat{\alpha}_{j}}. 
\end{equation*}
However, the asymptotic property of such estimator is unknown, and its finite sample performance might not be satisfactory. 
In addition, the asymptotic variance of the quantile treatment effect estimator becomes more complicated, as shown in the Corollary \ref{col:qte} below. 
Therefore, we propose an alternative inference method based on subsampling and establishes its asymptotic validity in Section \ref{sec:subsampling}.  
Conveniently, this method does not require estimating $c_j$ or other second-order parameters.

\section{Implications of the Main Theoretical Result}\label{sec:implication}

Given the estimator of $\alpha_j$ and its asymptotic behavior, we now proceed to estimate the extreme QTE.
Section \ref{sec:qte} presents our extreme QTE estimator and derives its asymptotic distribution.
However, the asymptotic distribution involves complicated variance expressions that vary across different applications (e.g., between Examples \ref{ex:abadie} and \ref{ex:frandsen}).
To overcome these challenges, Section \ref{sec:subsampling} proposes to conduct inference for the QTE based on subsampling as a general recipe in practice. 
A theoretical guarantee will be provided for this proposal.

\subsection{Extreme Quantile Treatment Effects}\label{sec:qte}

Our extreme QTE estimator can be written as
\begin{align}\label{eq:qte}
\widehat{QTE}(q) = 
y_{\min,1} \left( \frac{1-\hat\beta_1(y_{\min,1})}{1-q} \right)^{1/\hat\alpha_1}
-
y_{\min,0} \left( \frac{1-\hat\beta_0(y_{\min,0})}{1-q} \right)^{1/\hat\alpha_1}
\end{align}
The corollary below establishes its asymptotic distribution. 
Define
\begin{equation*}
\lambda _{jn}=\frac{A_{jn}^{-1/2}}{F_{j}^{-1}\left( 1-q\right) \log d_{jn}} \quad\text{ and }\quad \underline{\lambda }_{n}=\min \{\lambda _{1n},\lambda _{0n}\}
\end{equation*}
for $j=0,1,$
where 
\begin{equation*}
d_{jn}=\frac{1-\beta _{j}\left( y_{\min ,j}\right) }{q}.
\end{equation*}

\begin{corollary}\label{col:qte}
If Assumptions \ref{a:beta_hat}-\ref{a:y_bar} hold with $A_{jn}^{-1/2}B_{jn}\rightarrow 0$ and $\left( 1-\beta _{j}\left( y_{\min ,j}\right) \right) /q\rightarrow \infty $,
then we have
\begin{equation*}
\underline{\lambda }_{n}\left( \widehat{QTE}\left( q\right) -QTE\left( q\right) \right)\overset{d}{\rightarrow }\mathcal{N}\left( 0,\Omega _{Q}\right),
\end{equation*}
where the expression of $\Omega _{Q}$ is presented in \eqref{eq:Omega_Q} in the proof. 
\end{corollary}

\noindent
A proof is found in Appendix \ref{sec:col:qte}. 

The QTE involves the estimators $\hat{F}^{-1}_1(\cdot)$ and $\hat{F}^{-1}_0(\cdot)$, which converge at different rates respectively denoted by $\lambda_{1n}$ and $\lambda_{0n}$.  
Therefore, the convergence rate of $\widehat{QTE}$ is determined by the slower one, i.e., $\underline{\lambda }_{n}=\min \{\lambda _{1n},\lambda _{0n}\}$. 
Also, the asymptotic variance $\Omega _{Q}$ involves $\lim_{n\rightarrow\infty} \underline{\lambda }_{n}/\lambda_{1n}$ and $\lim_{n\rightarrow\infty} \underline{\lambda }_{n}/\lambda_{0n}$, which are bounded between zero and one, with the larger one being exactly one. 
In addition, it involves the covariance between $\hat{\alpha}_1$ and $\hat{\alpha}_0$, whose expression is complicated. 
We present the expression in the proof for brevity because we will not use it in practice thanks to the subsampling method presented in the following subsection.

\subsection{Inference}\label{sec:subsampling}

The asymptotic distribution of the QTE estimator is complex, making direct inference based on the analytically estimated $\Omega _{Q}$ challenging to implement. Additionally, the analytic expression for this distribution can vary across different contexts (e.g., between Examples \ref{ex:abadie} and \ref{ex:frandsen}), and no universal approach can be recommended for all applications. Moreover, previous studies have demonstrated that the bootstrap method performs poorly in estimating tail features \citep[e.g.,][]{hall1990using, bickel2008choice}. To address these challenges, we propose using subsampling as an alternative and establish its asymptotic validity.

Specifically, consider a sequence of subsample sizes $b=b_{n}$ that grows with $b/n\rightarrow 0$ as $n\rightarrow \infty $. Let $B_{n}=\binom{n}{b}$
denote the total possible number of subsamples of size $b$. 
For a given $b$ and $t\in \{1,...,b_{n}\}$, let $S_{t}\subset \{1,...,n\}$ be one of the $B_{n}$ subsamples of the individual indices with $|S_{t}|=b$, and define the estimator \eqref{eq:qte} based on the $t$-th subsample as $\hat{\alpha}_{j}^{t}$ for $j=0,1$. 
Furthermore, denote the corresponding QTE estimator based on the $t$-th subsample as
\begin{equation*}
\widehat{QTE}\left( q\right) ^{t}=y_{\min ,1}\left( \frac{1-\hat{\beta}_{1}\left( y_{\min ,1}\right) }{1-q}\right) ^{1/\hat{\alpha}_{1}^{t}}-y_{\min ,0}\left( \frac{1-\hat{\beta}_{0}\left( y_{\min ,0}\right) }{1-q}\right) ^{1/\hat{\alpha}_{0}^{t}}.
\end{equation*}
Note that the estimator $\hat{\beta}_{j}\left( \cdot \right) $ for the counterfactual CDF still uses the original full sample.  

Given $B_{n}$ subsampling estimates, we propose to use the empirical CDF
\begin{equation*}
L_{n,b}\left( s\right) =\frac{1}{B_{n}}\sum_{t=1}^{B_{n}}\mathbf{1}\left[ \underline{\lambda }_{b_n} \left( \widehat{QTE}\left( q\right) ^{t}-\widehat{QTE}\left( q\right) \right) \leq s\right] 
\end{equation*}
to approximate the CDF of $\widehat{QTE}\left( q\right) $, denoted by $L^{\ast }\left( \cdot \right) $.
The following corollary provides a theoretical guarantee that this subsampling approximation works asymptotically.

\begin{corollary}\label{col:subsampling}
Suppose that the conditions in Corollary 1 hold. If $b\rightarrow \infty $ and $b/n\rightarrow 0$ as $n\rightarrow \infty $, then we have
\begin{equation*}
\sup_{s\in \mathbb{R}}\left\vert L_{n,b}\left( s\right) -L^{\ast }\left( s\right) \right\vert \overset{p}{\rightarrow }0.
\end{equation*}
\end{corollary}

\noindent
A proof is found in Appendix \ref{sec:col:subsampling}.

In light of this theoretical result, we shall use the subsampling in the subsequent numerical and empirical analyses.








\section{Simulations}

In this section, we use simulated data to analyze the finite-sample performance of our proposed method.
We revisit Examples \ref{ex:abadie} and \ref{ex:frandsen} for simulation designs, and present them in Sections \ref{sec:simulations:abadie} and \ref{sec:simulations:rdd}, respectively.

\label{sec:simulations} 

\subsection{Example \ref{ex:abadie}: Binary Instrument and Binary Treatment}\label{sec:simulations:abadie}

Recall Example \ref{ex:abadie} from Section \ref{sec:main}.
In its setting, we generate independent copies of $(Y, D, Z, X' )'$ as follows.
An individual is an always taker, a complier, or a never taker, indicated by $\mathbbm{1}_A$, $\mathbbm{1}_C$, and $\mathbbm{1}_N$, respectively.
Each type emerges with the probability of $1/3$.
We generate $10$-dimensional exogenous covariates
$
X \sim \mathcal{N}(0,I_{10}),
$
where $I_{k}$ denotes the $k\times k$ identity matrix.
We in turn generate the instrument $Z \sim \text{Bernoulli}(\Lambda(X'\gamma))$, where $\Lambda$ is the logistic link function defined by $\Lambda(u) = \exp^u / (1+\exp^u)$ and $\gamma = (0.1,...,0.1)'$.
The treatment selection is generated by
\begin{align*}
D = \mathbbm{1}_A + \mathbbm{1}_C Z.
\end{align*}
The potential outcomes are given by 
\begin{align*}
Y_0 &= 
\mathbbm{1}_A \left(t_{0,A}-\frac{TE_A}{2}\right) +
\mathbbm{1}_C \left(t_{0,C}-\frac{TE_C}{2}\right) + 
\mathbbm{1}_N \left(t_{0,N}-\frac{TE_N}{2}\right)
\qquad\text{and}\\
Y_1 &= 
\mathbbm{1}_A \left(t_{1,A}+\frac{TE_A}{2}\right) +
\mathbbm{1}_C \left(t_{1,C}+\frac{TE_C}{2}\right) + 
\mathbbm{1}_N \left(t_{1,N}+\frac{TE_N}{2}\right),
\end{align*}
where
$TE_A=2$, $TE_C=1$, and $TE_N=0$ are the treatment effects for always takers, compliers, and never takers, respectively, and $t_{0,A}$, $t_{1,A}$, $t_{0,A}$, $t_{1,A}$, $t_{0,A}$ and $t_{1,A}$ are independently drawn from the Student-t distribution with 10 degrees of freedom.
Finally, the observed outcome is generated by
$$
Y = (1-D)Y_0 + D Y_1.
$$
In this manner we generate $n \in \{2500,5000,10000\}$ independent copies of $(Y,D,Z,X')'$.

Since lower extreme quantiles are often of policy interest, we focus on the QTE at $q \in \{0.01,...,0.05\}$.
We thus flip the sign of $Y$, and obtain the negative treatment effects.
The tuning parameter $y_{\min,j}$ is set as the 97.5-th percentile of $\hat\beta_j$.
We compute simulation statistics, such as the bias, standard deviation, root mean square error, and the 95\% coverage frequency based on 10,000 Monte Carlo iterations. 
Table \ref{tab:simlations:abadie} summarizes the simulation results.

\begin{table}
\centering
\begin{tabular}{rlrrrrr}
\multicolumn{7}{c}{Quantile Treatment Effects}\\
\hline\hline
&& \multicolumn{5}{c}{Quantile $q$}\\
\cline{3-7}
&& 0.005  & 0.010  & 0.015  & 0.020  & 0.025 \\
\hline
$n=2500$
& Bias & -6.020 & -0.022 & 0.007 & 0.001 & -0.008\\
& SD   & $>$10 & 2.530 & 0.401 & 0.369 & 0.373\\
& RMSE & $>$10 & 2.530 & 0.401 & 0.369 & 0.373\\
\cline{2-7}
& 95\% & 1.000 & 0.993 & 0.977 & 0.976 & 0.986\\
\hline
$n=5000$
& Bias & -0.014 & 0.008 & 0.007 & 0.002 & -0.005\\
& SD   & 0.503 & 0.326 & 0.275 & 0.259 & 0.257\\
& RMSE & 0.503 & 0.326 & 0.275 & 0.259 & 0.257\\
\cline{2-7}
& 95\% & 0.987 & 0.943 & 0.931 & 0.937 & 0.946\\
\hline
$n=10000$
& Bias & 0.024 & 0.026 & 0.017 & 0.008 & -0.001\\
& SD   & 0.347 & 0.231 & 0.195 & 0.183 & 0.180\\
& RMSE & 0.347 & 0.232 & 0.195 & 0.183 & 0.180\\
\cline{2-7}
& 95\% & 0.953 & 0.923 & 0.925 & 0.932 & 0.938\\
\hline\hline
\end{tabular}
\caption{Simulation results for Example \ref{ex:abadie}: binary instrument and binary treatment. The results are displayed for the quantiles $q \in \{0.005, 0.010, 0.015, 0.020, 0.025\}$ and sample sizes $n \in \{2500, 5000, 10000\}$. Displayed for each set of simulations are the bias, standard deviation (SD), root mean square error (RMSE), and 95\% coverage frequency.}${}$
\label{tab:simlations:abadie}
\end{table}

Observe that the RMSE diminishes as the sample size $n$ increases for each quantile.
Also, observe that the coverage frequency approach gets closer to the nominal probability of 95\% as the sample size $n$ increases for each quantile.
As the sample size increase, the coverage become more accurate for the extreme quantiles.

\subsection{Example \ref{ex:frandsen}: Regression Discontinuity Design}\label{sec:simulations:rdd}

Recall Example \ref{ex:frandsen} from Section \ref{sec:main}.
In its setting, we generate independent copies of $(Y,D,R)'$ as follows.
As in the previous subsection, an individual is an always taker, a complier, or a never taker, indicated by $\mathbbm{1}_A$, $\mathbbm{1}_C$, and $\mathbbm{1}_N$, respectively.
Each type emerges with the probability of $1/3$.
We generate the running variable
$
R \sim \mathcal{N}(0,1).
$
We in turn generate the treatment selection by
\begin{align*}
D = \mathbbm{1}_A + \mathbbm{1}_C \widetilde D(R),
\end{align*}
where $\widetilde D(R)\sim \text{Bernoulli}(1/3 + \mathbbm{1}\{R>0\}/3)$.
The potential outcomes are given by 
\begin{align*}
Y_0 &= 0.1 \cdot D +
\mathbbm{1}_A \left(t_{0,A}-\frac{TE_A}{2}\right) +
\mathbbm{1}_C \left(t_{0,C}-\frac{TE_C}{2}\right) + 
\mathbbm{1}_N \left(t_{0,N}-\frac{TE_N}{2}\right)
\quad\text{and}\\
Y_1 &= 0.1 \cdot D + 
\mathbbm{1}_A \left(t_{1,A}+\frac{TE_A}{2}\right) +
\mathbbm{1}_C \left(t_{1,C}+\frac{TE_C}{2}\right) + 
\mathbbm{1}_N \left(t_{1,N}+\frac{TE_N}{2}\right),
\end{align*}
where
$TE_A=2$, $TE_C=1$, and $TE_N=0$ are the treatment effects for always takers, compliers, and never takers, respectively, and $t_{0,A}$, $t_{1,A}$, $t_{0,A}$, $t_{1,A}$, $t_{0,A}$ and $t_{1,A}$ are independently drawn from the Student-t distribution with 10 degrees of freedom.
Finally, the observed outcome is generated by
$$
Y = (1-D)Y_0 + D Y_1.
$$
In this manner, we generate $n \in \{2500,5000,10000\}$ independent copies of $(Y,D,R)'$.

As in the previous example, we focus on the QTE at $q \in \{0.01,...,0.05\}$.
We thus flip the sign of $Y$, and obtain the negative treatment effects.
The tuning parameter $y_{\min,j}$ is set as the 97.5-th percentile of $\hat\beta_j$.
For the limits of the conditional expectation function, we use the one-sided Epanechnikov kernel with the rule-of-thumb bandwidth $h = \hat\sigma_R n^{-1/5}$, where $\hat\sigma_R^2$ denotes the sample variance of $R$.
We compute simulation statistics, such as the bias, standard deviation, root mean square error, and the 95\% coverage frequency based on 10,000 Monte Carlo iterations. 
Table \ref{tab:simlations:abadie} summarizes the simulation results.

\begin{table}
\centering
\begin{tabular}{rlrrrrr}
\multicolumn{7}{c}{Quantile Treatment Effects}\\
\hline\hline
&& \multicolumn{5}{c}{Quantile $q$}\\
\cline{3-7}
&& 0.005  & 0.010  & 0.015  & 0.020  & 0.025 \\
\hline
$n=2500$
& Bias & $<$-10 & $<$-10 & 0.080 & 0.109 & 0.095\\
& SD   & $>$10 & $>$10 & 1.493 & 0.440 & 0.428\\
& RMSE & $>$10 & $>$10 & 1.495 & 0.453 & 0.438\\
\cline{2-7}
& 95\% & 1.000 & 1.000 & 0.997 & 0.997 & 0.997\\
\hline
$n=5000$
& Bias & 0.095 & 0.119 & 0.114 & 0.105 & 0.096\\
& SD   & 0.906 & 0.488 & 0.357 & 0.314 & 0.307\\
& RMSE & 0.911 & 0.502 & 0.375 & 0.331 & 0.321\\
\cline{2-7}
& 95\% & 0.992 & 0.973 & 0.950 & 0.946 & 0.954\\
\hline
$n=10000$
& Bias & 0.143 & 0.135 & 0.119 & 0.104 & 0.091\\
& SD   & 0.646 & 0.373 & 0.277 & 0.243 & 0.234\\
& RMSE & 0.662 & 0.397 & 0.302 & 0.264 & 0.251\\
\cline{2-7}
& 95\% & 0.947 & 0.941 & 0.933 & 0.930 & 0.938\\
\hline\hline
\end{tabular}
\caption{Simulation results for Example \ref{ex:frandsen}: regression discontinuity design. The results are displayed for the quantiles $q \in \{0.005, 0.010, 0.015, 0.020, 0.025\}$ and sample sizes $n \in \{2500, 5000, 10000\}$. Displayed for each set of simulations are the bias, standard deviation (SD), root mean square error (RMSE), and 95\% coverage frequency.}${}$
\label{tab:simlations:frandsen}
\end{table}

The results are similar to those presented in the previous subsection.
Specifically, the RMSE diminishes as the sample size $n$ increases for each quantile.
Also, observe that the coverage frequency approach gets closer to the nominal probability of 95\% as the sample size $n$ increases for each quantile.
As the sample size increase, the coverage become more accurate for the extreme quantiles.

\section{Revealing the Causal Effects of the JTPA Program for the Most Disadvantaged Individuals}

\label{sec:application} 

In this section, we present an empirical application of our proposed method of estimation and inference for extreme quantile treatment effects.
We revisit the causal inference for the job training provided by the JTPA on participants' earnings and employment outcomes studied by numerous researchers \citep[e.g.,][]{HIST1998}.
Related to our approach,
\citet{abadie2002instrumental} used the method (Abadie's Kappa) presented in Example \ref{ex:abadie} to obtain selection-adjusted quantile regression estimates.
We use the same data set as in their empirical application.
Instead of following their estimation strategy, however, we use our proposed method of estimation and inference for extreme quantiles.

The data come from the National JTPA Study, which includes information on individuals who were randomly assigned to receive the JTPA training and those who were not.
Not all the assigned individuals took the treatment, and hence the compliance is imperfect.
Nonetheless, the random assignment serves as a natural instrument for the endogenous treatment.
We are interested in the treatment effects for the outcome of log wages.
In our analysis, we use the same set of covariates, as well as the outcome, treatment, and instrument as in the original study.

\begin{figure}[t]
\centering
\includegraphics[width=0.9\textwidth]{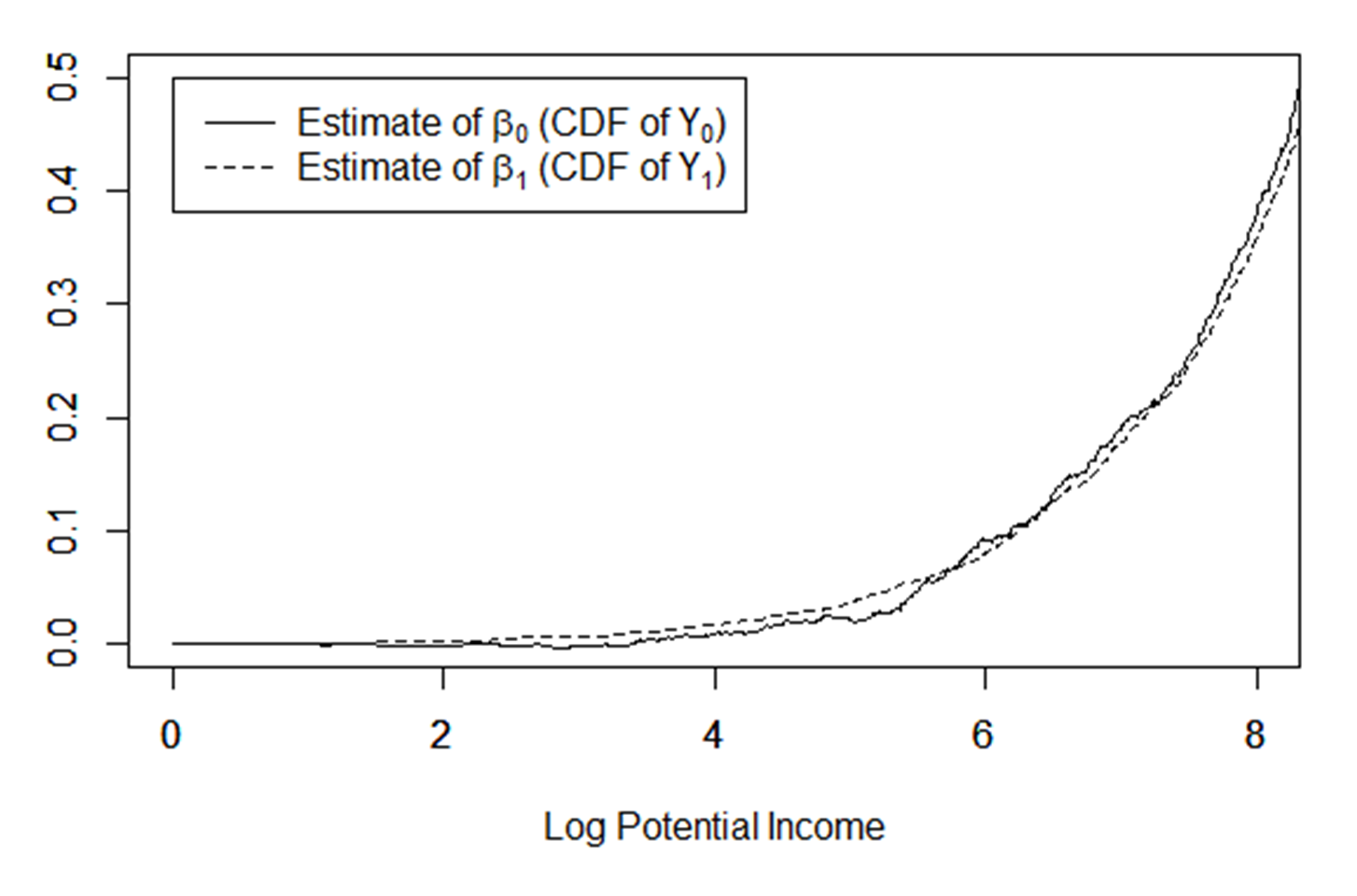}
\caption{Estimates $\hat\beta_0$ and $\hat\beta_1$ of the cumulative distribution functions $F_0$ and $F_1$ of the potential outcomes $Y_0$ and $Y_1$, respectively.}${}$
\label{fig:aii_beta_hat}
\end{figure}

Figure \ref{fig:aii_beta_hat} plots estimates, $\hat\beta_0$ and $\hat\beta_1$, of the distributions, $F_0$ and $F_1$, of the potential outcomes, $Y_0$ and $Y_1$, respectively.
Note that, unlike the empirical CDFs, these estimated CDFs do not need to be non-decreasing.
Observe that the estimate $\hat\beta_1$ dominates the estimate $\hat\beta_0$ for $q > 0.10$, implying positive quantile treatment effect estimates in these intermediate quantiles.
In contrast, the estimate $\hat\beta_0$ dominates the estimate $\hat\beta_1$ for $q < 0.05$, implying negative quantile treatment effect estimates in these extremely low quantiles.
However, the existing econometric methods of estimation and inference focusing on intermediate quantiles may not be applicable to analyzing the latter feature of these estimates.
Our proposed method can fill this gap.

Before conducting the QTE estimation, we first evaluate whether the Pareto-type tail condition (Assumption \ref{a:pareto}) is a suitable assumption for this dataset. 
Figure \ref{fig:paretofit} replicates the estimates, $\hat\beta_0$ (top) and $\hat\beta_1$ (bottom), as reported in Figure \ref{fig:aii_beta_hat}. 
Additionally, Figure \ref{fig:paretofit} also displays the Pareto fits of the corresponding distributions (depicted by gray lines), which are based on our estimates of the tail exponents, $\hat\alpha_0$ and $\hat\alpha_1$, for $\alpha_0$ and $\alpha_1$, respectively. 
By comparing the black and gray lines in each of the top and bottom panels, we can observe that the fit appears to be reasonably consistent. 

\begin{figure}
\centering
\includegraphics[width=0.9\textwidth]{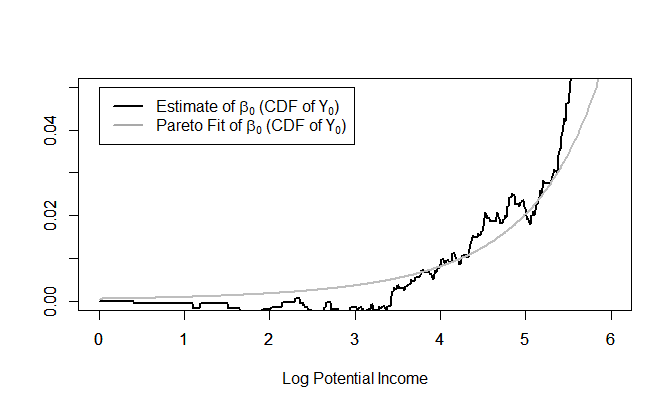}
\includegraphics[width=0.9\textwidth]{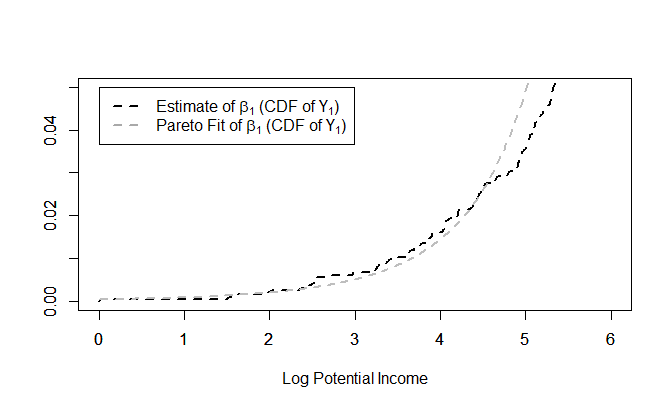}
\caption{Estimates $\hat\beta_0$ (top) and $\hat\beta_1$ (bottom) of the cumulative distribution functions $F_0$ and $F_1$ of the potential outcomes $Y_0$ and $Y_1$, respectively, together with their Pareto fits indicated by gray lines in the left tail of the distributions.}${}$
\label{fig:paretofit}
\end{figure}

We now estimate the extreme QTE using our proposed method. 
Figure \ref{fig:results} shows estimates and 95\% confidence intervals of the quantile treatment effects $QTE(q)$ for the extremely low quantiles $q \leq 0.025$.
Observe that the treatment effect is significantly \textit{negative} for each of $q \in [0.019, 0.025]$.
Hence, the job training program may exacerbate the labor outcomes for extremely disadvantaged individuals.
For even lower quantiles $q \in [0.002, 0.017]$, the point estimates are negative but the 95\% confidence intervals contain zero.

\begin{figure}
\centering
\includegraphics[width=0.9\textwidth]{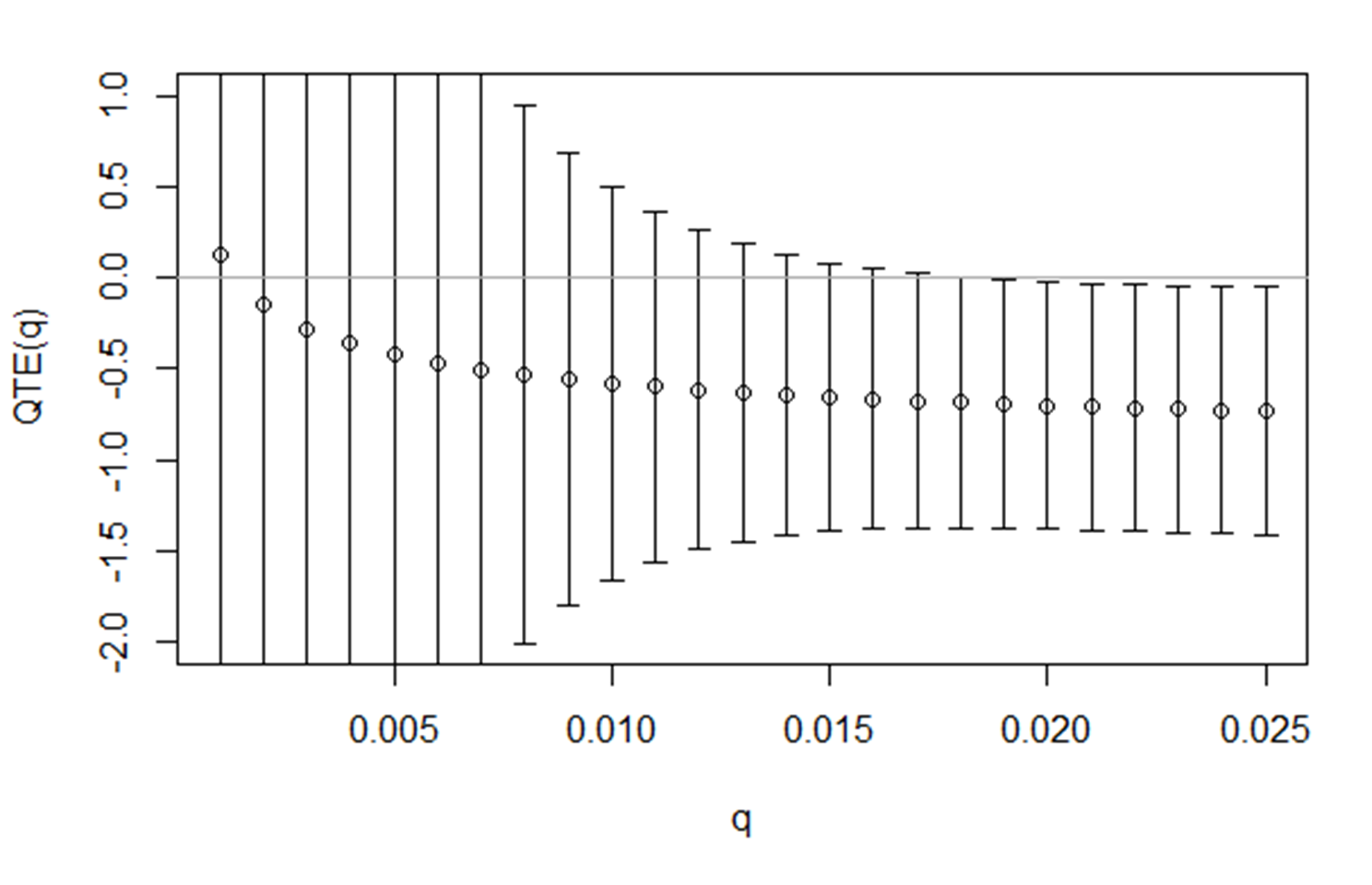}
\caption{Estimates of the $QTE(q)$ in the lower extreme quantiles and their 95\% confidence intervals. The horizontal axis measures the quantile $q$ and the vertical axis measures $QTE(q)$.}${}$
\label{fig:results}
\end{figure}

Policymakers are frequently concerned with understanding the impact of interventions on the most vulnerable populations, particularly those at the extreme lower end of the economic spectrum. Their objectives often include improving the economic well-being of the most disadvantaged individuals, who are typically represented by the lowest quantiles in a distribution.

However, traditional econometric methods have predominantly focused on intermediate quantiles, which, while useful, may not fully capture the nuances or the extent of impact on the extreme ends of the distribution. These methods, therefore, have inherent limitations when it comes to assessing treatment effects on the lowest quantiles, potentially overlooking critical insights that are essential for crafting effective policy interventions for the most at-risk groups.

In this empirical application, we demonstrate that an analysis centered on extremely low quantiles can yield surprising and sometimes counterintuitive results. Specifically, our study reveals that treatment effects, which policymakers might assume to be uniformly positive, could in fact be negative for these disadvantaged groups. This finding underscores the importance of employing methodologies that specifically target these extreme quantiles to ensure that policy decisions do not inadvertently exacerbate the conditions of those they are intended to help.


\section{Summary and Discussions}
\label{sec:conclusion} 

In this paper, we introduce a new method for the estimation and inference of extreme QTEs, which are crucial in understanding the impacts of interventions on the most vulnerable or advantaged individuals within a population. 
The proposed method is versatile and can be applied to a wide range of empirical research designs that have been developed for causal inference in the presence of endogeneity.
Leading examples are the instrumental variables design and the regression discontinuity design, which we focused on in this paper.
Our method leverages regular variation and subsampling. 
The combination of these two features ensures robust performance even in the sparsest regions of the distribution, where traditional methods might fail or produce unreliable results.
Our theoretical analysis is complemented by extensive simulation studies, which demonstrate that our method performs well under the scenarios of both the instrumental variables and regression discontinuity designs.

To illustrate the practical utility of our method, we apply it to a real-world policy evaluation: the assessment of job training programs provided by the JTPA. 
Previous studies using traditional instrumental variables estimation methods of the QTEs have largely focused on intermediate quantiles, finding positive treatment effects for individuals at the middle of the distribution. However, our method reveals a different picture when we extend the analysis to the extreme lower quantiles, representing the most disadvantaged individuals. Specifically, we find significantly negative QTEs for these individuals, suggesting that the job training program may not only fail to benefit the most vulnerable but might even have adverse effects on them. This finding contrasts sharply with the prevailing narrative in the literature, highlighting the importance of examining extreme quantiles to fully understand the impact of policy interventions across the entire distribution of outcomes.

Finally, we conclude by underscoring the theoretical contributions of this paper. Most existing methods for investigating tail features, including extreme quantiles, rely on \citeauthor{hill1975}'s (\citeyear{hill1975}) estimator or its numerous variants. See \citet{Hill2010,Hill2015}, among many others. These methods are effective when the sample is drawn directly from the target distribution, $F_j(\cdot)$, as is the case in the changes-in-changes design, where \citet{sasaki2024extreme} applied these approaches.
However, such straightforward methods are inadequate for other scenarios involving endogeneity, such as the instrumental variables design and the regression discontinuity design, where each observation cannot be uniquely associated with a specific distribution $F_j$, as illustrated in Examples \ref{ex:abadie} and \ref{ex:frandsen}. 
This limitation necessitated the development of a novel theoretical framework to study the tails of the distribution $F_j$ without direct observation of samples from $F_j$.
To our knowledge, this aspect of theoretical development is unprecedented in the literature, and it was driven by the need to address the challenges posed by specific empirical approaches that are widely used.

\vspace{1cm}
\section*{Appendix}
\appendix

The appendix collects mathematical proofs and technical details omitted from the main text.


\section{Proof of Theorem \protect\ref{thm:index}}

\label{sec:proof:index} 

\begin{proof}
Since we study $\hat{\alpha}_{j}$ separately for $j \in \{0,1\}$ in this theorem, we suppress the subscript $j$ throughout the proof to simplify notations. 
By \eqref{eq:w}, we have 
\begin{align*}
\hat{\alpha} =&-\frac{\int_{y_{\min}}^{\infty }\log \left( \frac{1-\hat{\beta}\left( y\right) }{1-\hat{\beta}\left( y_{\min}\right) }\right)
w\left( y\right) dy}{\int_{y_{\min}}^{\infty }\log \left( \frac{y}{y_{\min}}\right) w\left( y\right) dy} \\
=&-\omega ^{2}\int_{y_{\min}}^{\infty }\log \left( \frac{1-\hat{\beta}\left( y\right) }{1-\hat{\beta}\left( y_{\min}\right) }\right) \frac{y^{-\omega -1}}{y_{\min}^{-\omega }}dy.
\end{align*}
Also define
\begin{equation}\label{eq:alpha_tilde}
\tilde{\alpha}=-\omega ^{2}\int_{y_{\min}}^{\infty }\log \left( \frac{1-\beta \left( y\right) }{1-\beta \left( y_{\min}\right) }\right) \frac{y^{-\omega -1}}{y_{\min}^{-\omega }}dy,
\end{equation}
which is the infeasible estimator relying on the true CDF $\beta \left(\cdot \right) = F(\cdot)$ instead of the estimator $\hat\beta(\cdot)$. 
We decompose $\hat{\alpha}-\alpha $ as 
\begin{equation*}
\hat{\alpha}-\alpha =(\tilde{\alpha}-\alpha) + (\hat{\alpha}-\tilde{\alpha}),
\end{equation*}
where the first pair characterizes the bias and the second pair characterizes the stochastic part. 

We analyze the bias first. 
To this end, using $\beta \left( \cdot \right)=F\left( \cdot \right) $ and Assumption \ref{a:pareto}, we have
\begin{align}
\log \left( \frac{1-\beta \left( y\right) }{1-\beta \left( y_{\min}\right) }\right) 
=&\log \left( \left( \frac{y}{y_{\min}}\right) ^{-\alpha }\frac{1+dy^{-\rho }+o(y^{-\rho })}{1+dy_{\min}^{-\rho }+o(y_{\min}^{-\rho })}\right)  \nonumber\\
=&-\alpha \log \left( \frac{y}{y_{\min}}\right) +\log \left( 1+\frac{dy^{-\rho }-dy_{\min}^{-\rho }+o(y_{\min}^{-\rho })}{1+dy_{\min}^{-\rho }+o(y_{\min}^{-\rho })}\right)  \nonumber\\
=&-\alpha \log \left( \frac{y}{y_{\min}}\right) +d\left( y^{-\rho }-y_{\min}^{-\rho }\right) +o(y_{\min}^{-\rho })
\label{eq:log_1_beta}
\end{align}
for $y>y_{\min}$ as $\bar{y}_n \rightarrow \infty$.
Then, 
\begin{align*}
\tilde{\alpha}-\alpha  
=&\frac{\int_{y_{\min}}^{\infty }\left[ -\alpha \log \left( \frac{y}{y_{\min}}\right) +d\left( y^{-\rho }-y_{\min}^{-\rho }\right) +o(y_{\min}^{-\rho })\right] w\left( y\right) dy}{\int_{y_{\min}}^{\infty}-\log \left( \frac{y}{y_{\min}}\right) w\left( y\right) dy}-\alpha  \\
=&-d\omega ^{2}y_{\min}^{-\rho }\left[ \int_{y_{\min}}^{\infty}\left( \left( \frac{y}{y_{\min}}\right) ^{-\rho }-1\right) \frac{y^{-\omega -1}}{y_{\min}^{-\omega }}dy\right] \left( 1+o(1)\right)  \\
=&-d\omega ^{2}y_{\min}^{-\rho }\left[ \int_{1}^{\infty }\left( s^{-\rho}-1\right) s^{-\omega -1}ds\right] \left( 1+o(1)\right)  \\
=&-d\omega ^{2}y_{\min}^{-\rho}\left( \frac{1}{\rho +\omega }-\frac{1}{\omega }\right) \left(1+o(1)\right)\\
=&-B_n(1 + o(1)),
\end{align*}
where the first and second equalities are due to \eqref{eq:w}, \eqref{eq:w2}, \eqref{eq:alpha_tilde}, and \eqref{eq:log_1_beta},
the third equality is due to a change of variables, and
the fourth equality holds under Assumption \ref{a:y_bar}.

Next, we analyze the asymptotic behavior of the stochastic part, $\hat{\alpha}-\tilde{\alpha}$. 
Observe that
\begin{align}
\sup_{y \in [\bar{y}_n, \infty)} \left\vert \frac{\vert \hat\beta(y)-\beta(y) \vert }{y^{-\alpha}} \right\vert
\leq & \sup_{y \in \mathbb{R}} \left( \tilde{r}_n \frac{|\hat\beta(y)-\beta(y) |}{\Xi(y,y)^{1/2}}  \right) \left(  \tilde{r}_n^{-1} \sup_{y \in [\bar{y}_n, \infty) } y^{\alpha} \Xi(y,y)^{1/2} \right) \nonumber\\
 \leq & O_p(1) \times o(1),
 \label{eq:uniform_bound}
\end{align}
where the last line follows from Assumptions \ref{a:xi} and \ref{a:y_bar}.

Now, rewrite $\hat{\alpha}-\tilde{\alpha}$ as
\begin{align*}
\hat{\alpha}-\tilde{\alpha} 
=&\frac{-\int_{y_{\min}}^{\infty }\left[\log \left( \frac{1-\hat{\beta}\left( y\right) }{1-\beta \left( y\right) }\right) -\log \left( \frac{1-\hat{\beta}\left( y_{\min}\right) }{1-\beta\left( y_{\min}\right) }\right) \right] w\left( y\right) dy}{\int_{y_{\min}}^{\infty }\log \left( \frac{y}{y_{\min}}\right) w\left( y\right) dy}\\
=&\frac{-\int_{y_{\min}}^{\infty }\left[ \log \left( 1 - \frac{ \hat\beta(y)-\beta(y)}{1-\beta \left( y\right) }\right) -\log \left( 1 - \frac{\hat\beta(\bar y_n)-\beta(\bar y_n) }{1-\beta \left( y_{\min}\right) }\right) \right] w\left( y\right) dy}{\int_{y_{\min}}^{\infty }\log \left( \frac{y}{y_{\min}}\right) w\left( y\right) dy} \\
=&\frac{r_{n}^{-1}\int_{y_{\min}}^{\infty }\left( \frac{r_n(\hat\beta(y)-\beta(y))}{ cy^{-\alpha} }-\frac{r_n(\hat\beta(\bar y_n) - \beta(\bar y_n)) }{ c \bar{y}_n^{-\alpha} }\right) w\left( y\right) dy}{\int_{y_{\min}}^{\infty }\log \left( \frac{y}{y_{\min}}\right) w\left( y\right) dy}\left( 1+o_{p}(1)\right),
\end{align*}
where the last equality uses \eqref{eq:uniform_bound} and the fact that $\log \left( 1+x\right) =x+o(x)$ as $x\rightarrow 0$. 
 
Note that Assumption \ref{a:beta_hat} has
$$
r_{n}( \hat{\beta}\left( \cdot \right) -\beta\left( \cdot \right) ) \Rightarrow Z_{j}\left( \cdot \right),
$$
and
$$
\int_{y_{\min}}^{\infty }\log \left( \frac{y}{y_{\min}}\right) \frac{y^{-\omega -1}}{y_{\min}^{-\omega }}dy=\frac{1}{\omega ^{2}}
$$
holds by \eqref{eq:w}--\eqref{eq:w2}.
Therefore,
\begin{equation*}
A_{n}^{-1/2}\left( \hat{\alpha}-\tilde{\alpha}\right) \overset{d}{\rightarrow }\mathcal{N}\left( 0,1\right)
\end{equation*}
follows. 

Combine the above results to obtain
\begin{equation*}
A_{n}^{-1/2}\left( \hat{\alpha}-\alpha +B_{n}\right) \overset{d}{\rightarrow }\mathcal{N}\left( 0,1\right),
\end{equation*}
which completes the proof.
\end{proof}

\section{Proof of Corollary \protect\ref{col:qte}}\label{sec:col:qte}

\label{sec:proof:corollary1} 

\begin{proof}
Let us denote the two terms of the QTE estimator as
\begin{eqnarray*}
\widehat{QTE}\left( q\right) &=&y_{\min ,1}\left( \frac{1-\hat{\beta}_{1}\left( y_{\min ,1}\right) }{1-q}\right) ^{1/\hat{\alpha}_{1}}-y_{\min,0}\left( \frac{1-\hat{\beta}_{0}\left( y_{\min ,0}\right) }{1-q}\right) ^{1/\hat{\alpha}_{0}} \\
&\equiv &\hat{F}_{1}^{-1}\left( 1-q\right) -\hat{F}_{0}^{-1}\left(1-q\right) ,
\end{eqnarray*}
where $\hat{\beta}_{j}\left( \cdot \right) =\hat{F}_{j}\left(\cdot \right) $ and $\beta _{j}\left( \cdot \right) =F_{j}\left( \cdot\right) $. 
First we are going to show that
\begin{equation}
\left( 
\begin{array}{c}
\frac{A_{1n}^{-1/2}}{\log d_{1n}}\left( \frac{\hat{F}_{1}^{-1}\left(1-q\right) }{F_{1}^{-1}\left( 1-q\right) }-1\right) \\ 
\frac{A_{0n}^{-1/2}}{\log d_{0n}}\left( \frac{\hat{F}_{0}^{-1}\left(1-q\right) }{F_{0}^{-1}\left( 1-q\right) }-1\right)%
\end{array}%
\right) \overset{d}{\rightarrow }\mathcal{N}\left( 0,I_{2}\right) .
\label{eq:conv:quan}
\end{equation}%
Define 
\begin{equation*}
\hat{d}_{jn}=\frac{1-\hat{\beta}_{j}\left( y_{\min ,j}\right) }{q}\text{ and }d_{jn}=\frac{1-\beta _{j}\left( y_{\min ,j}\right) }{q}
\end{equation*}
for $j=0,1$.
Then we have
\begin{eqnarray*}
\left\vert \frac{\hat{d}_{jn}}{d_{jn}}\right\vert &\leq &\sup_{z}\left\vert \frac{1-\hat{\beta}_{j}\left( z\right) }{1-\beta _{j}\left( z\right) }\right\vert \leq 1+\sup_{z\in \mathbb{R}}\left\vert \frac{\hat{\beta}_{j}\left( z\right) -\beta _{j}\left( z\right) }{1-\beta _{j}\left( z\right) }\right\vert \\
&&\overset{(\ast )}{\leq }1+C\sup_{z\in \mathbb{R}}\left\vert \frac{\Xi _{j}\left( z,z\right) ^{1/2}\tilde{r}_{n}^{-1}}{z^{-\alpha _{j}}}\right\vert \rightarrow 1,
\end{eqnarray*}%
where the inequality marked by ($\ast $) follows from Assumptions \ref{a:beta_hat}--\ref{a:xi}, and the last convergence follows from Assumption \ref{a:y_bar}.

Now, substitute $\hat{d}_{nj}$ into $\hat{F}_{j}^{-1}\left( 1-q\right) $ to obtain
\begin{eqnarray*}
&&\frac{A_{jn}^{-1/2}}{\log d_{jn}}\left( \frac{\hat{F}_{j}^{-1}\left(1-q\right) }{F_{j}^{-1}\left( 1-q\right) }-1\right) \\
&=&\frac{d_{jn}^{1/\alpha _{j}}F_{j}^{-1}\left( \hat{F}_{j}(y_{\min,j})\right) }{F_{j}^{-1}\left( 1-q\right) }\left( \frac{A_{jn}^{-1/2}}{\log d_{jn}}\left( \frac{y_{\min ,j}}{F_{j}^{-1}\left( \hat{F}_{j}(y_{\min,j})\right) }-1\right) \hat{d}_{jn}^{1/\hat{\alpha}_{j}-1/\alpha _{j}}\right.\\
&&\left. +\frac{A_{jn}^{-1/2}}{\log d_{jn}}\left( \hat{d}_{jn}^{1/\hat{\alpha}_{j}-1/\alpha _{j}}-1\right) -\frac{A_{jn}^{-1/2}\hat{F}_{j}(y_{\min,j})^{-\rho _{j}}}{\log d_{jn}}\frac{\hat{d}_{jn}^{-1/\alpha _{j}}\frac{F_{j}^{-1}\left( 1-q\right) }{F_{j}^{-1}\left( \hat{F}_{j}(y_{\min,j})\right) }-1}{\hat{F}_{j}(y_{\min ,j})^{-\rho _{j}}}\right) \\
&\equiv &\Delta _{1n}\left( \Delta _{2n}+\Delta _{3n}-\Delta _{4n}\right) ,
\end{eqnarray*}%
where $A_{jn}$ is defined as in Theorem 1. We study each of the components, $\Delta _{1n},...,\Delta_{4n}$, one-by-one. 

For $\Delta _{1n}$, Assumption \ref{a:pareto} implies
\begin{equation*}
F_{j}^{-1}\left( t\right) \sim C_j \left( 1-t\right) ^{-1/\alpha _{j}}
\end{equation*}
for some constant $C_j>0$ as $t\rightarrow 1$. Thus,
\begin{equation*}
\Delta _{1n}\sim d_{jn}^{1/\alpha _{j}}\frac{\left( 1-\hat{F}_{j}(y_{\min,j})\right) ^{-1/\alpha _{j}}}{q^{-1/\alpha _{j}}}\overset{p}{\rightarrow }1.
\end{equation*}

For $\Delta _{2n}$, we have that 
\begin{align*}
& \text{ \ \ \ }\frac{y_{\min ,j}}{F_{j}^{-1}\left( \hat{F}_{j}(y_{\min,j})\right) }-1 \\
& \overset{(1)}{\sim }\frac{1}{f\left( F_{j}^{-1}\left( \hat{F}_{j}(y_{\min,j})\right) \right) }\frac{F_{j}\left( y_{\min ,j}\right) -\hat{F}%
_{j}(y_{\min ,j})}{F_{j}^{-1}\left( \hat{F}_{j}(y_{\min ,j})\right) } \\
& \overset{(2)}{=}\frac{1}{f\left( y_{\min ,j}\right) y_{\min ,j}}O_{p}\left( \Xi _{j}\left( y_{\min ,j},y_{\min ,j}\right) ^{1/2}\tilde{r}_{n}^{-1}\right) \\
& \overset{(3)}{=}O_{p}\left( y_{\min ,j}^{-\alpha _{j}}\Xi _{j}\left(y_{\min ,j},y_{\min ,j}\right) ^{1/2}\tilde{r}_{n}^{-1}\right) \\
& \overset{(4)}{=}o_{p}(1),
\end{align*}%
where the order equivalence (1) is due to the delta method; the equality (2) is by Assumption \ref{a:beta_hat}; the equality(3) is by Assumptions \ref{a:xi} and \ref{a:pareto}; and the equality (4) is by Assumption \ref{a:y_bar}. Furthermore, $\hat{d}_{jn}^{1/\hat{\alpha}_{j}-1/\alpha _{j}}\overset{p}{\rightarrow }1$ by the following argument for $\Delta _{3n}$. Then, $\Delta _{2n}=o_{p}(1)$ follows from $\log d_{jn}\rightarrow \infty $.

For $\Delta _{3n}$, by taking logarithm and the fact that $\exp x\sim 1+x$ for $x\rightarrow 1$, we obtain 
\begin{equation*}
\frac{A_{jn}^{-1/2}}{\log d_{jn}}\left( \hat{d}_{jn}^{1/\hat{\alpha}_{j}-1/\alpha _{j}}-1\right) \sim A_{jn}^{-1/2}\left( \frac{1}{\hat{\alpha}_{j}}-\frac{1}{\alpha _{j}}\right) \overset{d}{\rightarrow }\mathcal{N}\left( 0,1\right),
\end{equation*}%
where the convergence follows from the delta method, Theorem \ref{thm:index}, and the assumption that $A_{jn}^{-1/2}B_{jn}\rightarrow 0$.

For $\Delta _{4n}$, by Remark 3.2.7 in de Haan and Ferreira (2007), our Assumption \ref{a:pareto} satisfies the second-order condition as in equation (3.2.4) of \citet{de2006extreme}. 
In particular, their $A\left( t\right) \sim t^{-\rho _{j}}$ for $j=0,1$. Therefore, the second term in $\Delta _{4n}$ satisfies
\begin{equation*}
\frac{\hat{d}_{jn}^{-1/\alpha _{j}}\frac{F_{j}^{-1}\left( 1-q\right) }{F_{j}^{-1}\left( \hat{F}_{j}(y_{\min ,j})\right) }-1}{\hat{F}_{j}(y_{\min,j})^{-\rho _{j}}}=O(1)\text{.}
\end{equation*}%
Furthermore, $$A_{jn}^{-1/2}\hat{F}_{j}(y_{\min ,j})^{-\rho _{j}}=o_{p}(1)$$ provided that $A_{jn}^{-1/2}B_{jn}\rightarrow 0$. 
Thus, $\Delta_{4n}=o_{p}(1) $ follows from $\log d_{jn}\rightarrow \infty $.

Now, \eqref{eq:conv:quan} is established by combining the above bounds for $\Delta _{1n},...,\Delta _{4n}$. 
Recall the notation
\begin{equation*}
\lambda _{jn}=\frac{A_{jn}^{-1/2}}{F_{j}^{-1}\left( 1-q\right) \log d_{jn}},\text{ and }\underline{\lambda }_{n}=\min \{\lambda _{1n},\lambda _{0n}\}.
\end{equation*}
It follows from \eqref{eq:conv:quan} and the continuous mapping theorem that%
\begin{eqnarray*}
&&\underline{\lambda }_{n}\left( \widehat{QTE}\left( q\right) -QTE\left(q\right) \right) \\
&=&\frac{\underline{\lambda }_{n}}{\lambda _{1n}}\lambda _{1n}\left( \hat{F}_{1}^{-1}\left( 1-q\right) -F_{1}^{-1}\left( 1-q\right) \right) -\frac{\underline{\lambda }_{n}}{\lambda _{0n}}\lambda _{0n}\left( \hat{F}_{0}^{-1}\left( 1-q\right) -F_{0}^{-1}\left( 1-q\right) \right) \\
&=&\left( \frac{\underline{\lambda }_{n}}{\lambda _{1n}}A_{1n}^{-1/2}\left(\alpha _{1}^{-1}-\alpha _{1}^{-1}\right) -\frac{\underline{\lambda }_{n}}{\lambda _{0n}}A_{0n}^{-1/2}\left( \alpha _{0}^{-1}-\alpha _{0}^{-1}\right)\right) \left( 1+o_{p}(1)\right) \\
&\overset{d}{\rightarrow }&\mathcal{N}\left( 0,\Omega _{Q}\right) ,
\end{eqnarray*}%
where%
\begin{equation}\label{eq:Omega_Q}
\Omega _{Q}=\left( \lim_{n\rightarrow \infty }\frac{\underline{\lambda }_{n}}{\lambda _{1n}}\right) ^{2}+\left( \lim_{n\rightarrow \infty }\frac{\underline{\lambda }_{n}}{\lambda _{2n}}\right) ^{2}-2\lim_{n\rightarrow \infty }Cov\left( A_{1n}^{-1/2}\left( \hat{\alpha}_{1}^{-1}-\alpha_{1}^{-1}\right) ,A_{0n}^{-1/2}\left( \hat{\alpha}_{0}^{-1}-\alpha_{0}^{-1}\right) \right).
\end{equation}
It remains to simplify the covariance term. 
To this end, recall that in the proof of Theorem \ref{thm:index} with $A_{jn}^{-1/2}B_{jn}=o(1)$ for $j=0,1$, we have
\begin{align*}
\hat{\alpha}_j-\alpha_j = \omega^{2} r_{n}^{-1}\int_{y_{\min,j}}^{\infty }\left( \frac{ Z_j (y) }{ c_j y^{-\alpha_j} }-\frac{ Z_j (y_{\min,j}) }{ c_j y_{\min,j}^{-\alpha_j} }\right) w_j\left( y\right) dy    \left( 1+o_{p}(1)\right). 
\end{align*}
Then the delta method implies
\small
\begin{align*}
& Cov\left( \left( \hat{\alpha}_{1}^{-1}-\alpha_{1}^{-1}\right) , \left( \hat{\alpha}_{0}^{-1}-\alpha_{0}^{-1}\right) \right) \sim \;  \omega^{4} r_{n}^{-2} \times  \\
& Cov\left( \alpha_1^{-2} \int_{y_{\min,1}}^{\infty }\left( \frac{ Z_1 (y) }{ c_1 y^{-\alpha_1} }-\frac{ Z_1 (y_{\min,1}) }{ c_1 y_{\min,1}^{-\alpha_1} }\right) w_1\left( y\right) dy , 
                \alpha_0^{-2} \int_{y_{\min,0}}^{\infty }\left( \frac{ Z_0 (y) }{ c_0 y^{-\alpha_0} }-\frac{ Z_0 (y_{\min,0}) }{ c_0 y_{\min,0}^{-\alpha_0} }\right) w_0\left( y\right) dy \right).
\end{align*}
\normalsize
This completes a proof of the corollary.
\end{proof}

\section{Proof of Corollary \protect\ref{col:subsampling}}\label{sec:col:subsampling}
\label{sec:proof:corollary2} 
\begin{proof}
The proof follows from Theorem 2.2.1 in \citet{politis1999subsampling}. 
In particular, the limiting distribution $L^{\ast}$ is continuous and the subsampling size $b_n$ satisfies $b_n\rightarrow\infty$ and $b_n/n\rightarrow 0$. 
\end{proof}


\section{Details of Example \ref{ex:abadie}}

\label{sec:abadie_weak_convergence}

This section revisits Example \ref{ex:abadie} and establishes lower-level sufficient conditions for the weak convergence \eqref{eq:weak_abadie}.

Write $W=(Y,D,Z,X^{\prime })$ and define
\begin{align*}
b_0^N(y;W) =& 1\{Y < y \} \cdot \frac{(1-D)(p(X)-Z)}{(1-p(X))p(X)},
\\
b_1^N(y;W) =& 1\{Y < y \} \cdot \frac{D(Z-p(X))}{(1-p(X))p(X)},
\qquad\text{and}\\
b^D(W) =& 1 - \frac{D(1-Z)}{1-p(X)} - \frac{(1-D)Z}{p(X)}.
\end{align*}
Observe that we can write $$\beta_j(y) = E[b_j^N(y;W)]/E[b^D(W)]$$ for each $j \in \{0,1\}$.

Suppose that $p(X)$ is parameterized by $p(X;\gamma)$, and let $\hat\gamma$ denote an estimate of $\gamma$.
We can then estimate $b_0^N$, $b_1^N$, and $b^D$ by
\begin{align*}
\hat b_0^N(y;W) =& 1\{Y < y \} \cdot \frac{(1-D)(p(X;\hat\gamma)-Z)}{(1-p(X;\hat\gamma))p(X;\hat\gamma)},
\\
\hat b_1^N(y;W) =& 1\{Y < y \} \cdot \frac{D(Z-p(X;\hat\gamma))}{(1-p(X;\hat\gamma))p(X;\hat\gamma)},
\qquad\text{and}\\
\hat b^D(W) =& 1 - \frac{D(1-Z)}{1-p(X;\hat\gamma)} - \frac{(1-D)Z}{p(X;\hat\gamma)},
\end{align*}
respectively.
Thus,  $\beta_j(y)$ is estimated by
\begin{align*}
\widehat\beta_j(y) = \frac{\sum_{i=1}^n \hat b_j^N(y;W)}{\sum_{i=1}^n \hat b^D(W)}
\end{align*}
for each $j \in \{0,1\}$.
In this setup, we impose the following assumption.

\begin{assumption}\label{a:abadie}
(i)
$p(\cdot)$ is bounded away from zero and one on the support of $X$.
(ii)
$p(X) = p(X;\gamma)$.
(iii)
$p(X;\gamma)$ is twice continuously differentiable with respect to $\gamma$ with uniformly bounded derivatives.
(iv)
$\sqrt{n}\left(\hat\gamma-\gamma\right) = \frac{1}{\sqrt{n}}\sum_{i=1}^n \phi_\gamma(W_i) + o_p(1)$ for a measurable influence function $\phi_\gamma$ such that $E[\phi_\gamma(W)]=0$, $E[\phi_\gamma^2(W)]<\infty$, and $E[\phi_\gamma(W)^2 \cdot 1\{|\phi_\gamma(W)| > \eta \sqrt{n}\}] \rightarrow 0$ for every $\eta > 0$.
\end{assumption}

For instance, the maximum likelihood estimate $\hat\gamma = \arg \max_\gamma \sum_{i=1}^n \Big (Z_i \log\Lambda(X_i'\gamma) + (1-Z_i) \log(1-\Lambda(X_i'\gamma)) \Big)$ of the logit model $p(X;\gamma) = \Lambda(X'\gamma) = \exp(X'\gamma)/(1 + \exp(X'\gamma))$ yields the influence function
$\phi_\gamma(W) = E[X \Lambda(X'\Gamma)(1-\Lambda(X'\Gamma))X']^{-1} X(Z-\Lambda(X'\Gamma))$.
In this case, Assumption \ref{a:abadie} is satisfied provided that $X$ has a bounded $(2+\delta)$-th moment for some $\delta>0$.

In light of Assumption \ref{a:abadie} (iii), we let $\partial_{\gamma'} p(x;\gamma)$ denote the derivative of $p(x;\gamma)$ with respect to $\gamma'$, i.e., it is a row vector.
For instance, in the case of the logit model, the derivative is
$
\partial_{\gamma'} p(x;\gamma)
=
\Lambda'(x'\gamma) x'
=
\frac{e(x'\gamma)}{(1 + e(x'\gamma))^2} x'.
$

Let the derivatives of $b_j^N(y;W)$ for $j \in \{0,1\}$ with respect to $p(X)$ be denoted by
\begin{align*}
\partial_p b_0^N(y;W) =& 1\{Y < y \} \cdot (1-D) \cdot \Big(\frac{p(X)^2-Z(1-2p(X))}{(1-p(X))^2p(X)^2}\Big)
\qquad\text{and}\\
\partial_p b_1^N(y;W) =& 1\{Y < y \} \cdot D \cdot \Big(\frac{p(X)^2-Z(1-2p(X))}{(1-p(X))^2p(X)^2}\Big).
\end{align*}
Then, the influence function of
$$
\sqrt{n}\left(\frac{1}{n}\sum_{i=1}^n \hat b_j^N(y;W_i) - E[b_j^N(y;W)]\right) 
$$
can be written as
\begin{align*}
\psi_j^N(y;W_i) =  b_j^N(y;W_i) + E[\partial_p b_j^N(y;W) \partial_{\gamma'}p(X;\gamma)] \cdot \phi_\gamma(W_i)
\end{align*}
for each $ j \in \{0,1\}$, as formally established in the proof of Proposition \ref{prop:abadie}.
Similarly, let the derivative of $b^D(W)$ with respect to $p(X)$ be denoted by
\begin{align*}
\partial_p b^D(W) =& -\frac{D(1-X)}{(1-p(X))^2} + \frac{(1-D)Z}{p(X)^2}.
\end{align*}
Then, the influence function of
$$
\sqrt{n} \left( \frac{1}{n}\sum_{i=1}^N \hat b^D(W_i) - E\left[ b^D(W) \right] \right) =\mathbb{G}_n\psi^D(W) + o_p(1),
$$
can be written as
\begin{align*}
\psi^D(W_i) =  b^D(W_i) + E[\partial_p b^D(W) \partial_{\gamma'}p(X;\gamma)] \cdot \phi_\gamma(W_i),
\end{align*}
as formally established in the proof of Proposition \ref{prop:abadie}.
Thus, we obtain the following weak convergence result.

\begin{proposition}\label{prop:abadie}
Suppose that $\{W_i\}_i$ are i.i.d. copies of $W=(Y,D,Z,X^{\prime })$. Then,
\begin{align*}
\sqrt{n}\left( \widehat \beta _j(\cdot) - \beta_j(\cdot) \right) \Rightarrow \frac{1}{E[b^D(W)]} \zeta_j^N(\cdot) - \frac{E[b_j^N(\cdot;W)]}{E[b^D(W)]^2} \zeta^D,
\end{align*}
where $(\zeta_j^N(\cdot),\zeta^D)$ follows the zero-mean Gaussian process with covariance function
\begin{align*}
Cov(\zeta_j^N(y),\zeta_j^N(y')) =& E[\psi_j^N(y;W)\psi_j^N(y';W)] - E[\psi_j^N(y;W)] \cdot E[\psi_j^N(y';W)],
\\
Cov(\zeta_j^N(y),\zeta^D) =& E[\psi_j^N(y;W)\psi^D(W)] - E[\psi_j^N(y;W)] \cdot E[\psi^D(W)],
\end{align*}
$j \in \{0,1\}$, and $Var(\zeta^D) = Var(\psi^D(W))$.
\end{proposition}
\begin{proof}
We first focus on the numerator of the estimator.
We can write
\begin{align*}
&
\sqrt{n}\left( \frac{1}{n}\sum_{i=1}^n \hat b_j^N(y;W_i) - E[b_j^N(y;W)] \right)
\\
=&
\frac{1}{\sqrt{n}}\sum_{i=1}^n (\hat b_j^N(y;W_i) - b_j^N(y;W_i))
+
\mathbb{G}_n b_j^N(y;W)
\\
=&
\left(E[ \partial_p b_j^N(y;W) \partial_{\gamma'} p(X;\gamma) ] + o_p(1)\right) \cdot \sqrt{n} (\hat \gamma - \gamma)
+
\mathbb{G}_n b_j^N(y;W)
\\
=&
\left(E[ \partial_p b_j^N(y;W) \partial_{\gamma'} p(X;\gamma) ] + o_p(1)\right) \cdot \left( \mathbb{G}_n \phi_\gamma(W) + o_p(1)\right)
+
\mathbb{G}_n b_j^N(y;W),
\end{align*}
where
the second equality is due to Assumption \ref{a:abadie} (i)--(iii), and
the third equality is due to Assumption \ref{a:abadie} (iv).
Therefore, we obtain
\begin{align*}
\sqrt{n}\left(\frac{1}{n}\sum_{i=1}^n \hat b_j^N(y;W_i) - E[b_j^N(y;W)]\right) =& \mathbb{G}_n \psi_j^N(y;W) + o_p(1),
\end{align*}
for $j \in \{0,1\}$.
Similarly, we have
\begin{align*}
&\sqrt{n} \left( \frac{1}{n}\sum_{i=1}^N \hat b^D(W_i) - E\left[ b^D(W) \right] \right) =\mathbb{G}_n\psi^D(W) + o_p(1)
\end{align*}
for the denominator of the estimator.

Now, we are going to establish a weak convergence of $\mathbb{G}_n (\psi_j^N(\cdot; W), \psi^D(W))$ to the zero-mean Gaussian process $(\zeta_j^N(\cdot),\zeta^D)$ through four steps.
First, by Assumption \ref{a:abadie} (i), there exists $\underline p > 0$ such that $\underline{p} \leq p(\cdot) \leq 1-\underline{p}$.
By Assumption \ref{a:abadie} (iii), there exists $\overline q < \infty$ such that $\partial_{\gamma'}p(\cdot;\cdot) \leq \overline q$.
Define the function $$\overline \psi (\cdot) = 1+{\underline p}^{-1} (1- {\underline p})^{-1} + {\underline p}^{-2} (1- {\underline p})^{-2} \overline q \phi_\gamma(\cdot).$$
This $\overline \psi$ is an envelop of $\Psi_j^N = \{\psi_j^N(y;\cdot) : y \in \mathbb{R}\}$ such that $P^\ast (\Psi_j^N)^2 < \infty$ and $P^\ast (\Psi_j^N)^2 \{ |\Psi_j^N| > \eta \sqrt{n}\} \rightarrow 0$ for every $\eta>0$ under of Assumption \ref{a:abadie} (iv), where $P^\ast$ denotes the outer probability.
Second, observe that $$\Psi_{j,\delta}^N = \{\psi_j^1(y;\cdot)-\psi_j^1(y';\cdot):|y-y'|<\delta\}$$ and $(\Psi_{j,\delta}^N)^2$ are $P$-measureable for every $\delta>0$ under Assumption \ref{a:abadie} (iii)--(iv).
Third, 
\begin{align*}
&\sup_{|y-y'|<\delta_n} P(\psi_j^1(y;W)-\psi_j^1(y';W))^2 
\\
\leq& \delta_n {\underline p}^{-2} (1- {\underline p})^{-2} \Big[ 1 + 2 {\underline p}^{-1} (1- {\underline p})^{-1} \overline q P|\phi_\gamma| + {\underline p}^{-2} (1- {\underline p})^{-2} \overline q^2 P\phi_\gamma^2 \Big] \rightarrow 0
\end{align*}
as $\delta_n \rightarrow 0.$
Fourth, the VC-index of $\Psi_j^N = \{\psi_j^N(y;\cdot) : y \in \mathbb{R}\}$ is 2.
To see this consider set $\{w_1,w_2\} = \{(y_1,d,z,x'),(y_2,d,z,x')\}$ with $y_1< y_2$.
Observe that the subgraph of an element of $\Psi_j^N$ cannot pick out $\{w_2\}$.
Therefore, it follows from \citet[][Theorem 2.11.22]{van1996weak} that
\begin{align*}
&\sqrt{n} \left(\begin{array}{cc} \frac{1}{n}\sum_{i=1}^n \hat b_j^N(y;W_i) - E\left[ b_j^N(y;W) \right] \\ \frac{1}{n}\sum_{i=1}^n \hat b^D(W_i) - E\left[ b^D(W) \right] \end{array}\right) \Rightarrow \left(\begin{array}{c}\zeta_j^N(\cdot) \\ \zeta^D\end{array}\right) \qquad\text{ for } j \in \{0,1\},
\end{align*}
where $(\zeta_j^N(\cdot),\zeta^D)$ follows the zero-mean Gaussian process with covariance function
\begin{align*}
Cov(\zeta_j^N(y),\zeta_j^N(y')) =& E[\psi_j^N(y;W)\psi_j^N(y';W)] - E[\psi_j^N(y;W)] \cdot E[\psi_j^N(y';W)],
\\
Cov(\zeta_j^N(y),\zeta^D) =& E[\psi_j^N(y;W)\psi^D(W)] - E[\psi_j^N(y;W)] \cdot E[\psi^D(W)],
\end{align*}
for $j \in \{0,1\}$, and $Var(\zeta^D) = Var(\psi^D(W))$.

Therefore,
\begin{align*}
\sqrt{n}\left( \widehat \beta _j(\cdot) - \beta_j(\cdot) \right) \Rightarrow \frac{1}{E[b^D(W)]} \zeta_j^N(\cdot) - \frac{E[b_j^N(\cdot;W)]}{E[b^D(W)]^2} \zeta^D,
\end{align*}
$j \in \{0,1\}$,
follows by the functional delta method.
\end{proof}

\vspace{1cm}

\bibliographystyle{ecta}
\bibliography{mybib}
\end{document}